\documentclass{llncs}
\usepackage{graphicx,cite,hyperref}
\usepackage{amsmath,amsfonts}
\usepackage{microtype}
\usepackage{wrapfig}

\let\doendproof\endproof
\renewcommand\endproof{~\hfill$\qed$\doendproof}

\newcommand{\NP}{\ensuremath{\mathsf{NP}}}
\DeclareMathOperator{\MSO}{MSO}

\begin{document}
\title{Rooted Cycle Bases}

\author{David Eppstein\inst{1} \and J. Michael McCarthy\inst{2} \and Brian E. Parrish\inst{2}}

\institute{Computer Science Department, University of California, Irvine. \and
Department of Mechanical and Aerospace Engineering,\\ University of California, Irvine.}

\date{ }
\maketitle

\thispagestyle{empty}
\pagestyle{plain}

\begin{abstract}
A cycle basis in an undirected graph is a minimal set of simple cycles whose symmetric differences include all Eulerian subgraphs of the given graph.
We define a rooted cycle basis to be a cycle basis in which all cycles contain a specified root edge, and we investigate the algorithmic problem of constructing rooted cycle bases. We show that a given graph has a rooted cycle basis if and only if the root edge belongs to its 2-core and the 2-core is 2-vertex-connected, and that constructing such a basis can be performed efficiently. We show that in an unweighted or positively weighted graph, it is possible to find the minimum weight rooted cycle basis in polynomial time.  Additionally, we show that it is \NP-complete to find a fundamental rooted cycle basis (a rooted cycle basis in which each cycle is formed by combining paths in a fixed spanning tree with a single additional edge) but that the problem can be solved by a fixed-parameter-tractable algorithm when parameterized by clique-width.
\end{abstract}

\section{Introduction}

\emph{A cycle basis} of an undirected graph is a set of cycles such that all cycles in the graph have a unique representation as an algebraic sum of basis cycles. In this paper we study algorithms for finding a special type of cycle basis which we call a \emph{rooted cycle basis}, in which all cycles in the basis contain a specified root edge.

Cycle bases  have diverse applications including subway system scheduling~\cite{Lie-ORP-07}, the analysis of distributed algorithms~\cite{BouPetVil-PODC-04}, and bioinformatics~\cite{AguIst-JCB-12,LemMaj-NAR-06}. The specific motivation for our rooted variant of the problem comes from mechanical engineering, where cycle bases
have long been used in static analysis of structures such as truss bridges~\cite{Kav-CMAME-76} and in the kinematics of moving bodies~\cite{KecKruHil-MSD-97}. We recently used this method as part of a system for constructing the configuration space of moving linkages~\cite{ParMcCEpp-JMM-15}, systems that include automobile suspensions, fold-out sofa-beds, and legs for walking robots.

In this configuration space construction problem, systems of rigid two-dim\-en\-sional \emph{links} 
are connected at \emph{joints} where one link can rotate around a point of another with one degree of freedom. A system of links and joints is called a \emph{kinematic chain}; fixing the position of one \emph{ground} link results in a system called a \emph{mechanism} or \emph{inversion}, and distinguishing a second \emph{input link} (connected to the ground by a joint and to which force is applied to control the rest of the system) results in a system called a \emph{linkage}~\cite{ParMcCEpp-JMM-15}. The structure of a linkage can be expressed combinatorially by a \emph{linkage graph}, an undirected graph with a vertex for each link and an edge for each joint, including a distinguished ground-input edge. The requirement that the combined motion of the linkage have one degree of freedom can be expressed combinatorially by the property that the linkage graph is \emph{$(\frac32,2)$-tight}~\cite{LeeStr-DM-08}: every $k$-vertex induced subgraph must have at most $\frac32k-2$ edges, and the whole graph must have exactly $\frac32n-2$ edges, where $n$ is the number of vertices in the graph and links in the linkage. Links may cross each other in the plane, resulting in a non-planar linkage graph (\autoref{fig:nonplanar}).

\begin{figure}[t]
\centering\includegraphics[width=0.8\textwidth]{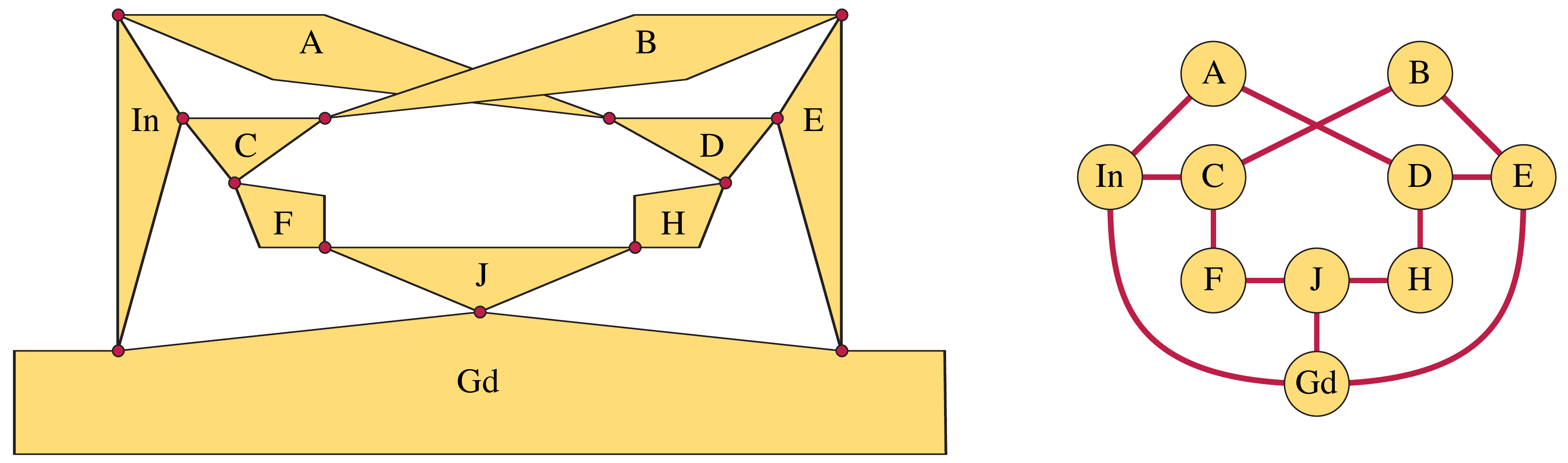}
\caption{A linkage and its linkage graph (a subdivision of $K_{3,3}$)}
\label{fig:nonplanar}
\end{figure}

Given a linkage with its linkage graph, each input-to-ground path has an associated equation representing the requirement that the joints along the path have angles consistent with the fixed ground position at both ends of the path. Our system for constructing the configuration space of a linkage chooses a complete and non-redundant subset of these path equations and uses Dixon determinants to solve this system of equations~\cite{ParMcCEpp-JMM-15}. Each path determining one of these equations can be turned into a cycle by adding the input-ground edge, and a set of equations chosen in this way is complete and non-redundant if and only if the corresponding set of cycles forms a cycle basis of the linkage graph. However, all of these cycles contain the input-ground edge, so the system of equations that we seek comes from a rooted cycle basis. Additionally, because the equation solver forms the computational bottleneck of our system, we would like the system of equations that we construct to be as simple as possible, corresponding to the problem of finding a \emph{minimum rooted cycle basis}.

\paragraph{New results.}
We provide the first algorithmic study of the problem of constructing rooted cycle bases.
We have the following new results:
\begin{itemize}
\item As a warm-up to our main result, we show that an arbitrary graph $G$ with designated root edge~$e$ has a rooted cycle basis in which all cycles contain $e$ if and only if the 2-core of $G$ is 2-vertex-connected and contains~$e$. When a rooted cycle basis exists, it can be constructed in time $O(mn)$. This is tight: there exist graphs for which every rooted cycle basis has total size $\Theta(mn)$.
\item Our main result is that, in an unweighted or positively weighted graph with a designated root edge, we can find the minimum weight rooted cycle basis by a randomized algorithm with nearly-optimal $O(mn+n^2\log n)$ expected time or by a polylogarithmically slower deterministic algorithm. This basis is always \emph{weakly fundamental}: its cycles can be ordered so that each cycle contains an edge that is not in any earlier cycle. Our algorithm uses a greedy method for finding each cycle, with a tie-breaking rule that avoids greedy choices that do not lead to a valid cycle basis.
\item In Appendix~\ref{sec:duality} we show that it is \NP-complete to determine whether a graph $G$ with root edge~$e$ has a \emph{fundamental} rooted cycle basis, a rooted cycle basis determined from a spanning tree $T$ by choosing all cycles formed by an edge not in $T$ and a path in~$T$. It remains \NP-complete even when $G$ is planar. Our proof is based on the observation that, in planar graphs, fundamental rooted cycle bases are dual to a form of Hamiltonian cycle. Additionally, we use Courcelle's theorem to show that finding a fundamental rooted cycle basis is fixed-parameter-tractable in the clique-width of the input.

\end{itemize}

In comparison, for arbitrary cycle bases, every graph has a fundamental cycle basis, which may be constructed using any spanning tree algorithm. Finding unrestricted minimum weight cycle bases takes polynomial time~\cite{AmaIulRiz-IPCO-10,Hor-SJC-87,MehMic-JEA-06,KavLieMeh-CSR-09}. However, finding an unrestricted minimum weight weakly fundamental cycle basis is \NP-hard~\cite{Riz-Algo-09}, and cannot be solved by the same greedy strategy that we use for rooted cycle bases, of choosing the shortest cycle that includes a new~edge.

\section{Preliminaries}

By $\mathbb{F}_2$ we mean the field with two elements $0$ and $1$ under mod-$2$ arithmetic. If $U$ is an arbitrary finite set, the subsets of $U$ form a vector space $\mathbb{F}_2^U$ over $\mathbb{F}_2$ with the empty set as origin and the symmetric difference of sets as addition.

We define a \emph{rooted graph} to be an undirected graph $G=(V,E)$ with a designated root edge~$e$.
A \emph{cycle} is a connected 2-regular subgraph; a cycle is \emph{rooted} if it contains $e$, and \emph{Hamiltonian} if it contains every vertex of $G$. The \emph{edge space} of~$G$ is the vector space $\mathbb{F}_2^E$. The \emph{cycle space} of $G$ is the subspace of the edge space generated by edge sets of cycles; its elements are  subgraphs of $G$ with even degree at every vertex~\cite{Tut-DM-71}. A \emph{cycle basis} of $G$ is a set of cycles that forms a basis of the cycle space~\cite{KavLieMeh-CSR-09}. A cycle basis is \emph{rooted} if all its cycles are rooted.

A \emph{spanning tree} of an undirected graph $G$ is a subgraph that includes all vertices of $G$, and is connected with no cycles. Any edge $f$ that does not belong to a spanning tree $T$ gives rise to a \emph{fundamental cycle} for $T$ consisting of $f$ plus the unique path in~$T$ connecting the endpoints of~$f$. The fundamental cycles for $T$ form a cycle basis; a  basis formed in this way is called \emph{fundamental}.

A \emph{matroid}~\cite{Wel-MT-10} may be defined as a family of subsets of a finite set, called the \emph{independent sets} of the matroid, with two properties:
\begin{itemize}
\item Every subset of an independent set is independent.
\item If $I_1$ and $I_2$ are independent sets and $|I_1|<|I_2|$, then there exists an element $x$ belonging to $I_2\setminus I_1$ such that $I_1\cup\{x\}$ is independent.
\end{itemize}
The linearly independent subsets of a finite family of vectors in any vector space form a \emph{linear matroid}.
In a matroid, a \emph{basis} is an independent set all of whose supersets are dependent; for linear matroids, this notion coincides with the standard definition of a basis of a vector space.

If the elements of a matroid are given real-valued weights, then the basis with minimum total weight can be constructed by a greedy algorithm, generalizing Kruskal's algorithm for minimum spanning trees: initialize a set $\mathcal{B}$ to be the empty set, and consider the elements in sorted order by their weights, adding each element to $\mathcal{B}$ if the result would remain independent. In particular, if the edges of an undirected graph $G$ are given weights, the weight of a cycle may be defined as the sum of the weights of its edges, and the weight of a cycle basis may be defined as the sum of the weights of its cycles. Then the minimum weight cycle basis may be found by considering all of the cycles of the graph in sorted order by weight, adding each one to the basis if the result would remain independent. This algorithm may be sped up by considering only a special set of polynomially-many candidate cycles, leading to polynomial-time construction of the minimum weight cycle basis in any graph~\cite{AmaIulRiz-IPCO-10,Hor-SJC-87,MehMic-JEA-06,KavLieMeh-CSR-09}.

A \emph{simple path} in a graph $G$ is a connected subgraph with two degree-one vertices (its endpoints) and with all remaining vertices (its interior vertices) having degree exactly two. An \emph{open ear decomposition} of $G$ is a collection of simple paths $P_i$ for $i=0,1,2,\dots$ (called ears) with the following properties:
\begin{itemize}
\item The first ear $P_0$ is a single edge.
\item The two endpoints of each ear $P_i$ with $i>0$ appear in earlier-numbered ears.
\item No interior vertex of an ear appears in any earlier ear.
\end{itemize}
A graph has an open ear decomposition if and only if it is $2$-vertex-connected (no vertex deletion can disconnect the remaining graph)~\cite{Whi-TAMS-32}. This decomposition can be constructed in linear time, with any edge as the first ear~\cite{Lov-FOCS-85,Tar-SJC-72,Ram-SoPA-93}.
The number of ears equals one plus the dimension of the cycle space. 

A vertex of $G$ belongs to at least one cycle of $G$ if and only if it belongs to the \emph{$2$-core} of~$G$, the subgraph formed by removing isolated vertices and degree-one vertices until all remaining vertices have degree $\ge 2$. Therefore, the cycle bases of $G$ are the same as the cycle bases of its $2$-core.

\section{Existence and construction of rooted cycle bases}

The following lemma is a special case of Menger's theorem, but we give a proof as we use the proof construction in our algorithms.

\begin{lemma}
\label{lem:2vc-paths}
Let $e$ be an edge of a $2$-vertex-connected graph $G$. Then for every two distinct vertices $u$ and $v$ of $G$ there exist two vertex-disjoint paths (possibly of length zero) from $u$ and $v$ respectively to the two endpoints of $e$.
\end{lemma}

\begin{proof}
Let $P_0=e, P_1,\dots, P_k$ be an open ear decomposition of~$G$. We apply induction on~$k$, with the following cases:
\begin{itemize}
\item As a base case, if $k=0$, we have two length-zero paths, one for each endpoint.
\item If $k>0$ and neither $u$ nor $v$ is an interior vertex of $P_k$, the result follows by induction on the union of the ears up to $P_{k-1}$.
\item If $k>0$ and exactly one of $u$ or $v$ is an interior vertex of $P_k$, without loss of generality (by swapping $u$ and $v$ if necessary) we may assume that $u$ is the interior vertex.
At least one endpoint of $P_k$ is a vertex $w$ distinct from $v$.
By induction, $v$ and $w$ can be connected by vertex-disjoint paths to $e$, using only vertices in ears $P_0,\dots,P_{k-1}$. The result follows by augmenting the path from $w$ with the part of path $P_k$ from $u$ to $w$.
\item If $k>0$ and both $u$ and $v$ are interior vertices of $P_k$, then $u$ and $v$ have two disjoint paths within $P_k$ to the endpoints of $P_k$. By induction, the endpoints of $P_k$ can be connected by paths to $e$, using only vertices in ears $P_0,\dots,P_{k-1}$.  The result follows by concatenating these paths with the paths within $P_k$.
\end{itemize}
Thus, in all cases, the desired two paths exist.
\end{proof}

An ear with one edge cannot be part of a path constructed by this proof. So for a graph $G$ with $n$ vertices and $m$ edges and a known ear decomposition, we can discard the one-edge ears and transform the case analysis of the proof into an algorithm
that constructs the two desired paths in time $O(n)$.

\begin{theorem}
An undirected graph $G$ rooted at edge $e$ has a cycle basis that is rooted at $e$ if and only if $e$ belongs to the $2$-core of $G$ and the $2$-core is $2$-vertex-connected. When a rooted cycle basis exists, it can be constructed in time $O(mn)$ and the total length of the cycles in the basis is $O(mn)$.
\end{theorem}

\begin{proof}
If $G$ has a rooted cycle basis, its $2$-core must be $2$-vertex-connected. For, suppose that a vertex $v$ is deleted from the $2$-core. Every remaining vertex $u$ belongs to a basis cycle from which only $v$ can have been deleted, leaving a path connecting $u$ to the remaining endpoints of~$e$. In this way any two remaining vertices can be connected to each other via~$e$, so the remaining vertices are not disconnected.

In the other direction, suppose that the $2$-core of $G$ contains $e$ and is $2$-connected. Then it has an open ear decomposition $P_0=e, P_1,\dots, P_k$. We may form a set of cycles $C_1,C_2,\dots, C_k$ in which each cycle $C_i$ consists of $e$, the edges in $P_i$, and two paths through the union of ears $P_1,P_2,\dots P_{i-1}$ connecting the endpoints of $P_i$ to the endpoints of $e$. These cycles are independent because each one contains at least one edge in $P_i$ that does not belong to any previous cycle. As an independent set of cycles of the correct cardinality to be a basis, they must be a basis.

After computing the ear decomposition, each cycle takes time $O(n)$ to construct (by the remarks following lemma 1) and has length $O(n)$, giving the stated time and length bounds.
\end{proof}

\begin{figure}[t]
\centering\includegraphics[width=\textwidth]{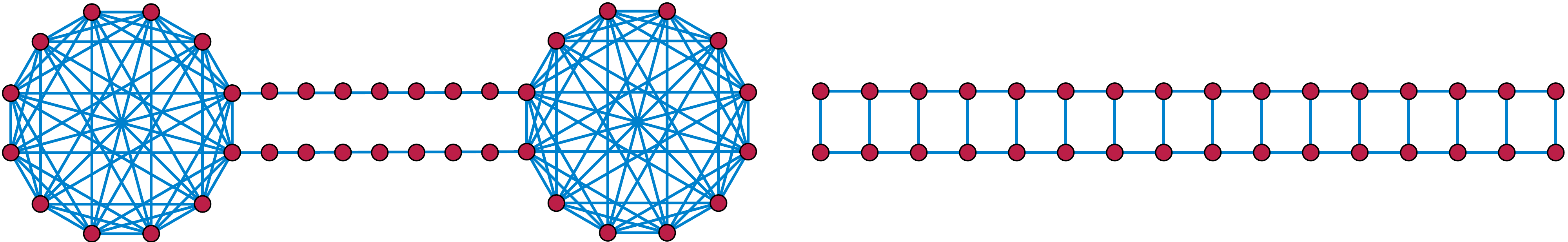}
\caption{Two graphs whose rooted cycle bases all have large total length: two cliques connected by two long paths (left), and a ladder graph (right).}
\label{fig:big-rooted-bases}
\end{figure}

The length and time  bounds of the theorem are tight in the worst case: for a graph consisting of two  $\Theta(m)$-vertex cliques connected by two $\Theta(n)$-vertex paths  (\autoref{fig:big-rooted-bases}, left), every cycle through $e$ and an edge in the farthest clique from $e$ has length $n$, so
every rooted cycle basis has total length $\Theta(mn)$. For linkage graphs with $m=\frac{3}{2}n-2$,  the time and length bounds become $O(n^2)$, which is again tight: every rooted cycle basis of an $n$-vertex ladder graph (\autoref{fig:big-rooted-bases}, right) has total length $\Theta(n^2)$.

In contrast, unrooted cycle bases may be significantly smaller. Every graph with $m$ vertices and $n$ edges has an (unrooted) cycle basis of total length $O(\min(n^2,m\log n))$, a bound that is close to tight because of the existence of sparse graphs of high girth for which every cycle basis has total length $\Omega(n\log n)$~\cite{ElkLieRiz-IPL-07,KavLieMeh-CSR-09}.

\section{Finding the minimum weight rooted cycle basis}

In this section we show how to find a rooted cycle basis of minimum total length in biconnected graphs with positive edge weights, in polynomial time.
We use a greedy algorithm that chooses one cycle at a time, and prove it correct by showing that the sequence of cycles selected by this algorithm correspond to an ear decomposition.
Our strategy is to show that an optimal basis  can be derived from an ear decomposition: the cycles of the basis can be sorted from shorter to longer cycles in such a way that, in each successive cycle, the edges that do not belong to earlier cycles form an ear.
Our algorithm performs the following steps:
\begin{enumerate}
\item Initialize what will eventually become a cycle basis to the empty set.
\item Use Suurballe's algorithm (reviewed in Appendix~\ref{sec:suurballe}) to compute, for each edge, the shortest rooted cycle through that edge.
\item While there exists an edge that is not included in any of the cycles chosen so far, select an edge that has not yet been included and whose computed shortest-cycle length is as small as possible, and add its cycle to the basis.
\end{enumerate}

\begin{figure}[b]
\centering\includegraphics[height=0.8in]{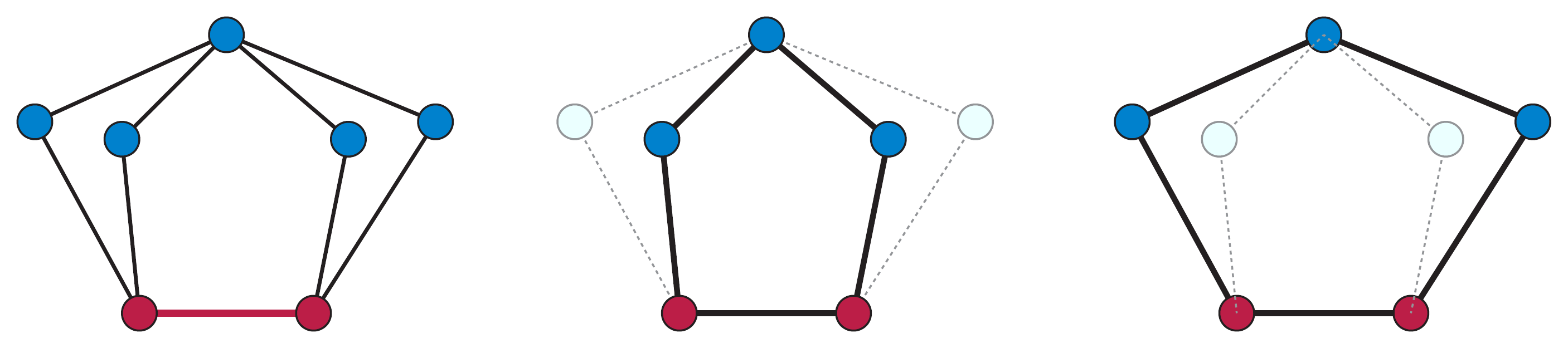}
\caption{An unweighted rooted graph (left) with two shortest rooted cycles that together cover the whole graph but do not generate its cycle  space (center and right). Our algorithm requires that no two cycles have equal weight, to prevent bad sets of cycles such as these from being chosen.}
\label{fig:bad-greedy}
\end{figure}

We will prove this algorithm correct under the additional assumption that no two paths, and no two cycles, have the same weight as each other. We say that a graph is \emph{unambiguously weighted} when this is the case. When paths and cycles can have equal weights, this algorithm can fail by choosing a set of cycles that together cover all edges but do not generate the whole cycle space (\autoref{fig:bad-greedy}), so we need a consistent tie-breaking rule in this case, which we describe in Appendix~\ref{sec:unambiguous}.

\subsection{Greedy cycle sequences}

We define a \emph{greedy cycle sequence} to be a sequence of cycles that could be produced by the algorithm described at the beginning of this section. That is, it is a sequence of rooted cycles $C_1$, $C_2$, $\dots$ in which
\begin{enumerate}
\item Each cycle includes an edge that is not in any earlier cycle in the sequence, and
\item Subject to constraint (1), each cycle is as short as possible.
\end{enumerate}

We will prove a sequence of lemmas about greedy cycle sequences, with the goal of showing that the set of new edges added by each cycle forms an ear and therefore that our greedy algorithm for rooted cycle bases is correct. To do so, it is helpful to have a notation for the subgraph of~$G$ formed by the vertices and edges in the first~$i$ cycles in the sequence. We call this subgraph the $i$th \emph{ambit} of the cycle basis,
\begin{wrapfigure}[20]{r}{1.65in}
\vspace{-3ex}
\centering
\includegraphics[scale=0.45]{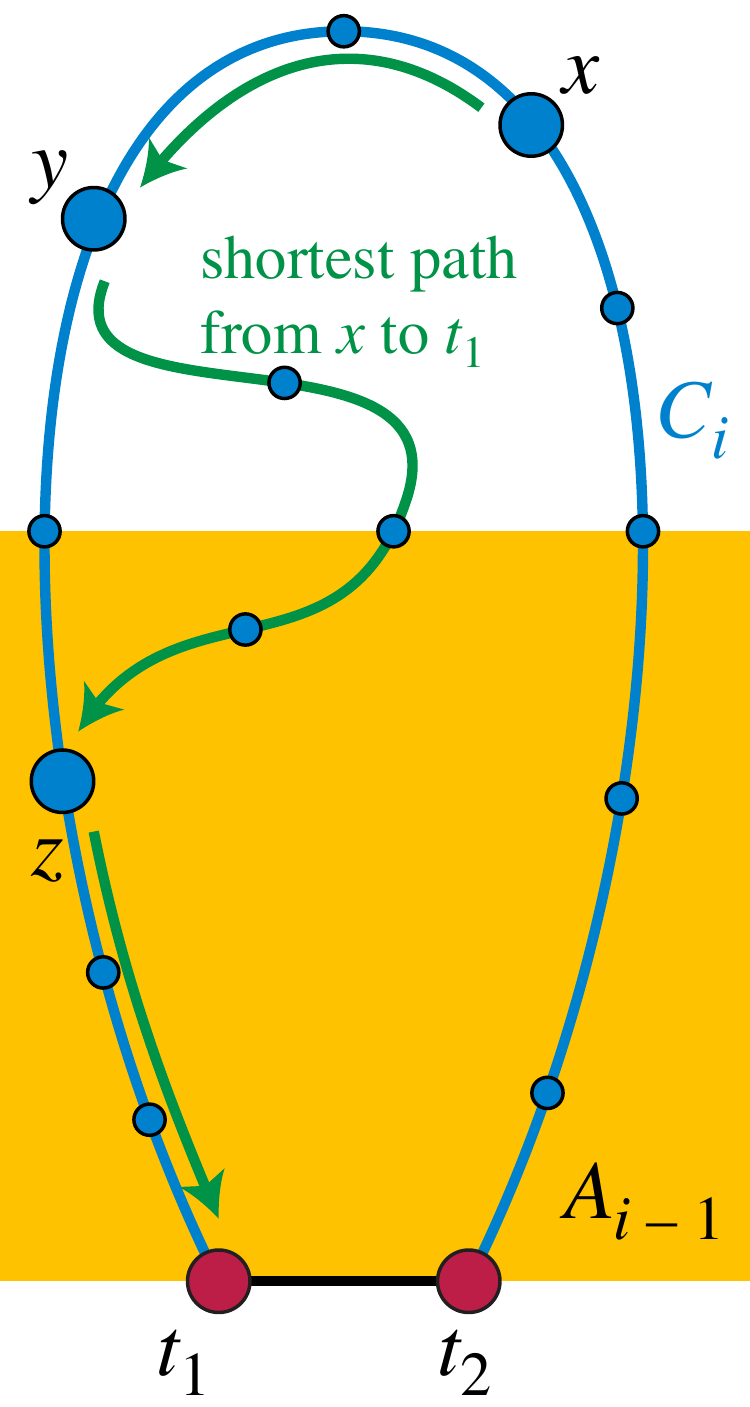}
\vspace{-1ex}
\caption{Notation for \autoref{lem:ambit-contains-path}. A shortest path is shown diverging from $C_i$ before reaching $A_{i-1}$, proven impossible by the lemma.}
\label{fig:ambit-contains-path}
\end{wrapfigure}
denoted $A_i$.

\begin{lemma}
\label{lem:ambit-contains-path}
Let $G$ be an unambiguously-weighted rooted graph, $A_i$ be the $i$th ambit of a greedy cycle sequence $C_1$, $C_2$, $\dots$ for $G$, and $x$ be a point in $A_i$ (a vertex or a point interior to an edge). Then $A_i$ contains the shortest path in $G$ from $x$ to each endpoint of the root edge of $G$.
\end{lemma}

\begin{proof}
Let $t_1$ and $t_2$ be the endpoints of the root edge. We will show by induction on~$i$ that $A_i$ contains the shortest path from $x$ to $t_1$; by symmetry it also contains a path to~$t_2$.
We may assume that $x$ does not belong to $A_{i-1}$, for otherwise the shortest path is already contained in $A_{i-1}$ by induction.

Then let $P$ be the shortest path in $G$ from $x$ to $t_1$; we claim that $P$ must remain within $C_i$ until it reaches a vertex of $A_{i-1}$. For, if $P$ deviated from $C_i$ at some vertex $y$ outside of $A_{i-1}$, let $z$ be the first point at which $P$ returns to a vertex of $C_i$; $z$ must exist, because $P$ eventually reaches $t_1$, which belongs to $C_i$. In this case the rooted cycle formed from $C_i$ by removing the path in $C_i$ from $y$ to $z$ and replacing it with the part of $P$ from $y$ to $z$ would be strictly shorter than $C_i$ (because the part of $P$ from $y$ to $z$ is a shortest path and no two paths have equal length) and would contain the vertex $y$ outside $A_{i-1}$, contradicting the greedy choice of $C_i$ as the shortest rooted cycle not contained in $A_{i-1}$.

Therefore $A_i$ contains the portion of $P$ from $x$ to $A_{i-1}$, and by induction it contains as well the rest of the shortest path from $x$ to $t_1$.
\end{proof}

\subsection{Rungs of the Suurballe ladder}

Let $P_1$ and $P_2$ be two disjoint paths from $s$ (an arbitrary vertex in the given graph $G$) to $t_1$ and $t_2$ (the endpoints of the root edge of $G$), as constructed by Suurballe's algorithm.
Recall that this algorithm constructs two different paths from $s$ to $t_1$ and $t_2$, the first of which is a shortest path in $G$ and the second of which is the shortest path in a derived graph $H$. The union of these two paths differs from $P_1$ and $P_2$ by a collection of paths that we call \emph{rungs}. Each rung is traversed in one direction by the shortest path in $G$ and in the opposite direction by the shortest path in $H$. The endpoints of the rungs lie on $P_1$ and $P_2$ (in the same order on both paths) and each of the two shortest paths is formed by following one of $P_1$ or $P_2$ until reaching the endpoint of a rung, then traversing that rung and continuing in the same way along $P_2$ or $P_1$ until the next rung, etc.
\begin{wrapfigure}[20]{r}{1.65in}
\vspace{-2ex}
\centering
\includegraphics[scale=0.4]{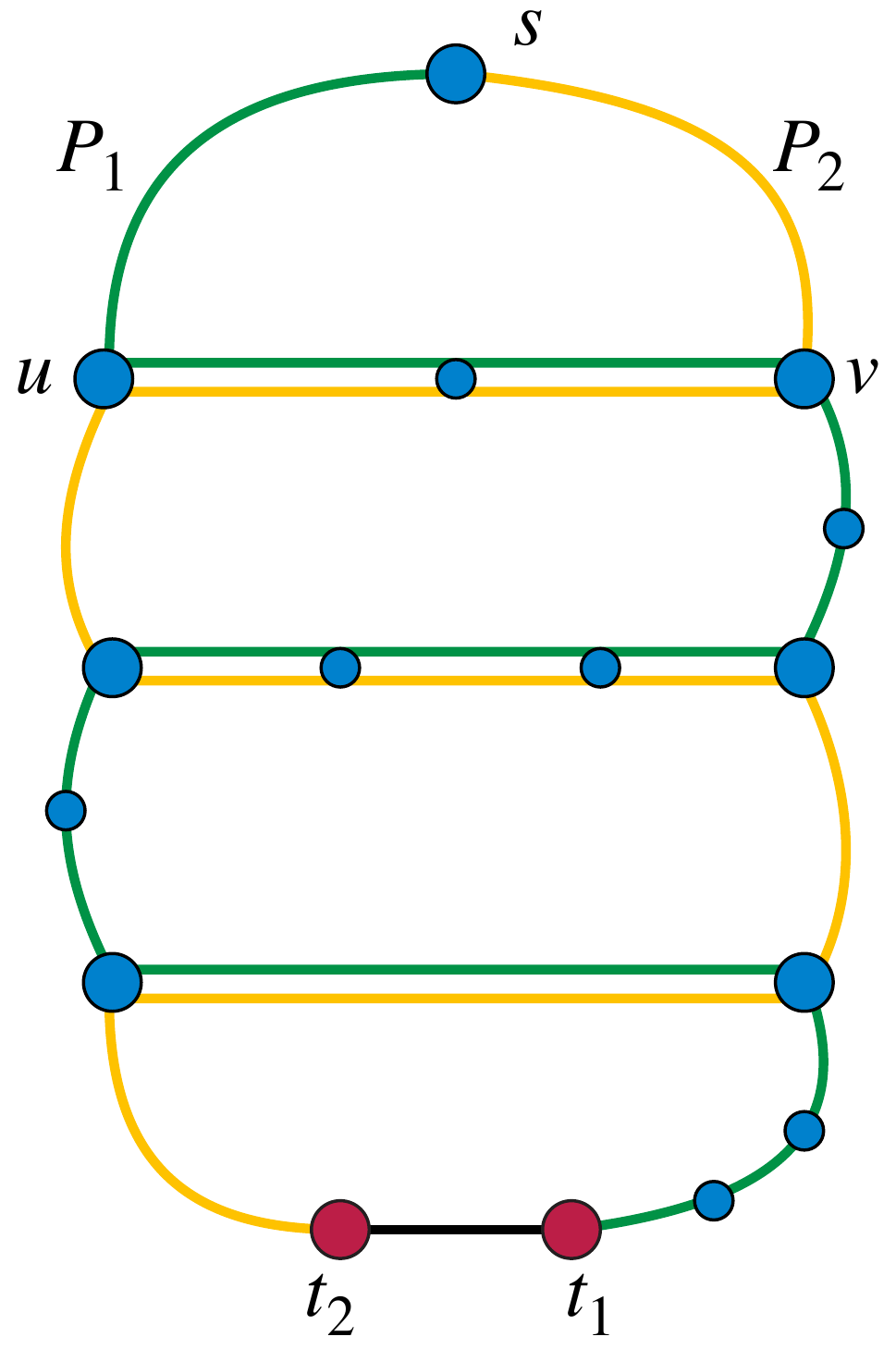}
\vspace{-1ex}
\caption{Notation for \autoref{lem:remove-rung}.  $P_1$ and $P_2$ are shown in contrasting colors; the horizontal segments with both colors are the rungs.}
\label{fig:rungs}
\end{wrapfigure}

\begin{lemma}
\label{lem:remove-rung}
With $s$, $P_1$, and $P_2$ as above, let $C$ be the cycle formed by $P_1$, $P_2$, and the root edge, and let  $R$ be the rung of $P_1$ and $P_2$ closest to $s$. Let $D$ be the cycle formed by removing the parts of $P_1$ and $P_2$  from $s$ to the endpoints of $R$, and replacing them with $R$.
Then $C$ is longer than $D$.
\end{lemma}

\begin{proof}
Let $u$ and $v$ be the endpoints of $R$, let $\ell_{su}$ denote the length of the path in $P_1$ from $s$ to $u$, let $\ell_{sv}$ denote the length of the path in $P_2$ from $s$ to $v$, and let $\ell_{uv}$ denote the length of the rung. $P_1$ follows $C$ from $s$ to $u$  then crosses rung~$R$; by  construction, it is the shortest path in $G$ from $s$ to $t_1$. Thus, $\ell_{su}+\ell_{uv}\le \ell_{sv}$, for if not then $P_1$ couldn't be a shortest path. Equivalently, adding $\ell_{su}-\ell_{uv}$ to both sides of the inequality gives $2\ell_{su}\le \ell_{su}+\ell_{sv}-\ell_{uv}$.
But the left hand side of this inequality is positive (by the assumption that the input graph has positive edge weights) and the right hand side is the difference in weights between $C$ and $D$.
\end{proof}

\begin{corollary}
\label{cor:rung-removal}
Let $C_i$ be a cycle in the greedy cycle sequence, and let $A_{i-1}$ be the ambit of the previous cycle. Suppose that $C_i$ is constructed by applying Suurballe's algorithm from a starting vertex $s$, and suppose that the two disjoint paths $P_1$ and $P_2$ comprising $C_i$ have a nonempty sequence of rungs. Then these rungs, and all parts of $P_1$ and $P_2$ from the first rung endpoint to $t_1$ and $t_2$, belong to $A_{i-1}$.
\end{corollary}

\begin{proof}
Otherwise the cycle $D$ in the statement of \autoref{lem:remove-rung}, or one of the other cycles constructed in the same way from one of the other rungs, would be a shorter cycle containing at least one edge that is not in $A_{i-1}$, and would have been selected in place of $C_i$ in the greedy cycle sequence.
\end{proof}

\subsection{From cycle sequences to ear decompositions}

As we now show, a greedy cycle sequence with cycles $C_1$, $C_2$, $\dots$ and ambits $A_1$, $A_2$, $\dots$ may be used to derive an ear decomposition, in which the first ear $P_0$ is the root edge and each subsequent ear $P_i$ is the subgraph $A_i\setminus A_{i-1}$. That is, for each $i$, this subgraph is a single path.

\begin{figure}[t]
\centering\includegraphics[scale=0.4]{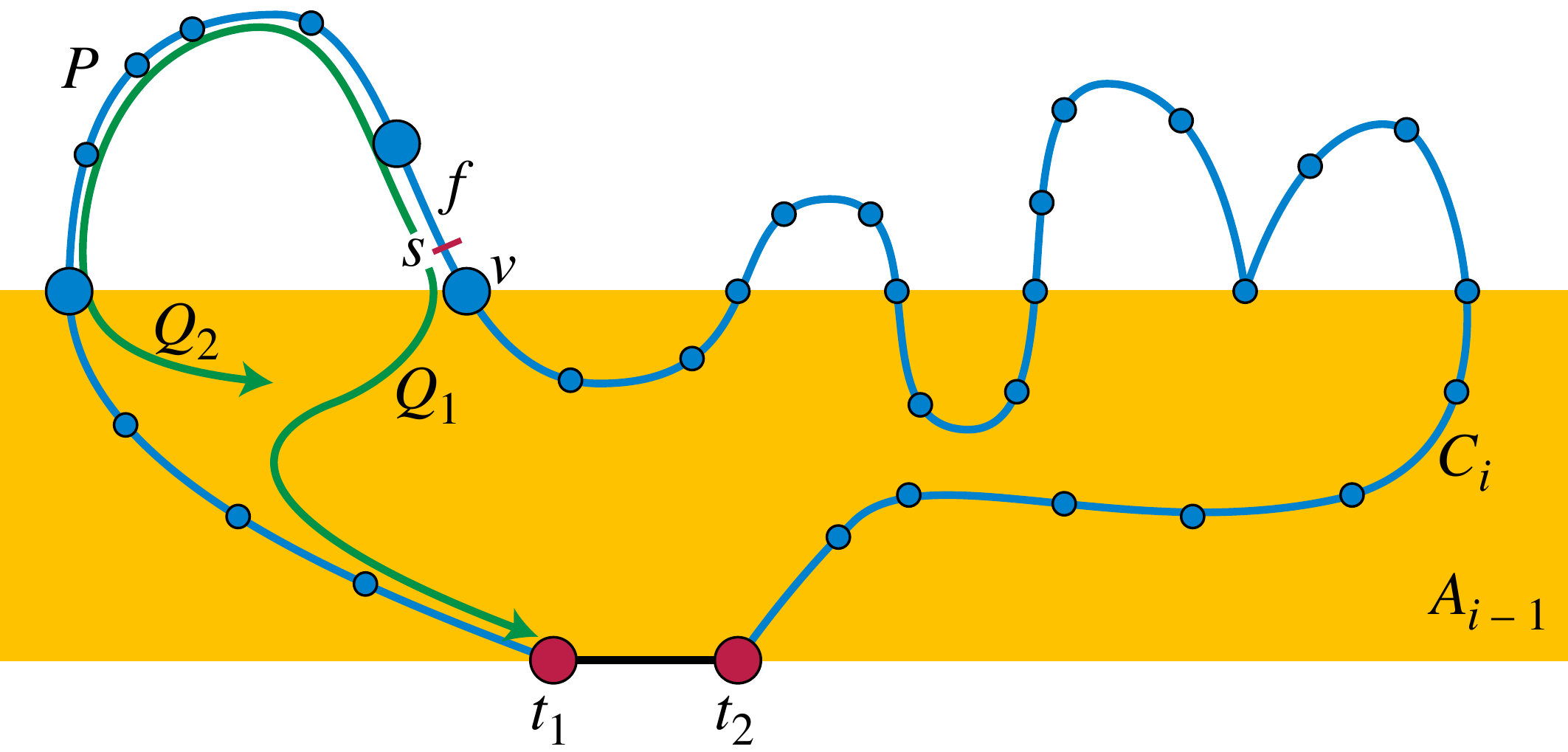}
\caption{Notation for \autoref{lem:cycles2ears}. The figure shows a cycle $C_i$ such that $C_i\setminus A_{i-1}$ forms more than one path, a configuration that is proven impossible by the lemma.}
\label{fig:one-ear}
\end{figure}

\begin{lemma}
\label{lem:cycles2ears}
Let $C_i$ and $A_i$ be a greedy cycle sequence and the corresponding sequence of ambits in an unambiguously-weighted graph. Then $A_i\setminus A_{i-1}$ forms a single connected path in the given graph~$G$.
\end{lemma}

\begin{proof}
Let $t_1$ and $t_2$ be the endpoints of the root edge.
As with any rooted cycle not  contained in $A_{i-1}$, $C_i$ forms one or more connected paths in $A_i\setminus A_{i-1}$, separated by vertices or edges of $A_{i-1}$. Let $P$ be the closest to $t_1$ of these paths according to their ordering along $C_i$, and let $f$ and $v$ be the farthest edge and vertex from $t_1$ in the path ordering of~$P$. Then $v$ is an endpoint of $f$ and of path $P$, and belongs to $A_{i-1}$. By \autoref{lem:ambit-contains-path} the shortest path from $v$ to $t_1$ stays within $A_{i-1}$ and therefore does not use edge~$f$.

Let $s$ be a point on edge $f$, sufficiently close to $v$ that the shortest path from $s$ to~$t_1$ passes through $v$. (In other words, subdivide the edge, and place a vertex at this point.) If we apply Suurballe's algorithm starting from the point $s$, it will find two paths $Q_1$ (the shortest path from $s$ to $t_1$) and $Q_2$. The symmetric difference of these two paths is $C_i$ again, because $C_i$ is the unique shortest rooted cycle containing~$s$. $Q_1$ passes from $s$ through $v$ and then stays within $A_{i-1}$ by \autoref{lem:ambit-contains-path}. $Q_2$ follows the rest of $P$, then stays within $A_{i-1}$ until it reaches the endpoint of the first rung of $Q_1\cap Q_2$ (by the choice of $P$ as the first of the paths in $A_i\setminus A_{i-1}$ in cycle order). After reaching this rung endpoint, $Q_2$ continues to stay within $A_{i-1}$ by \autoref{cor:rung-removal}. Thus, neither $Q_1$ nor $Q_2$ can escape $A_{i-1}$ once they enter it, so the cycle $C_i$ that they form can only have the single component $P$ outside of~$A_{i-1}$.
\end{proof}

\begin{corollary}
\label{cor:ear}
Let $C_i$ be the cycles of a greedy cycle sequence for the rooted graph $G$, and let $A_i$ be the corresponding ambits. Let $P_0$ be the one-edge path formed by the root edge of $G$, and for $i>0$ let $P_i$ be the path $A_i\setminus A_{i-1}$. Then the sequence of paths
$P_0,P_1,\dots P_i$ is an ear decomposition of the subgraph $A_i$, and the sequence of all paths formed in this way is an ear decomposition of~$G$.
\end{corollary}

\begin{proof}
By \autoref{lem:cycles2ears}, each of these graphs is a path; its endpoints belong to earlier paths and its edges and interior vertices do not. Thus, this sequence of paths satisfies all the requirements of an ear decomposition.
\end{proof}

\subsection{Greed is Good}

We have nearly completed the proof of correctness of our greedy algorithm for constructing minimum weight rooted cycle bases.

\begin{lemma}
\label{lem:optimal4ambit}
Let $C_i$ be the cycles of a greedy cycle sequence for an unambiguously-weighted rooted graph $G$, and let $A_i$ be the corresponding ambits.
Then the set of cycles $C_1,C_2,\dots C_i$ is a minimum weight rooted cycle basis for $A_i$.
\end{lemma}

\begin{proof}
We use induction on $i$. For $i=1$, the graph $C_1=A_1$ has only the one cycle. For $i>1$, this set of cycles is linearly independent because each contains at least one edge not found in earlier cycles. The number of cycles is the same as the number of ears (after the root edge) in an ear decomposition, by \autoref{cor:ear}, so it equals the dimension of the cycle space. As an independent set of the correct number of cycles, these cycles must form a cycle basis for $A_i$.

Because the cycle space of any graph forms a matroid, the minimum weight basis of any subset $\mathcal{S}$ of cycles can be found by a greedy algorithm that at each step selects the minimum-weight member of $\mathcal{S}$ that is independent of previous selections. By the induction hypothesis, any cycle that is independent of the previous $i-1$ selections must use at least one edge outside of $A_{i-1}$, and $C_i$ is the minimum-weight rooted cycle with this property. Therefore, the cycles form a minimum-weight rooted cycle basis.
\end{proof}

\begin{theorem}
The minimum weight rooted cycle basis of a biconnected rooted graph $G$ with positive edge weights can be constructed in polynomial time.
\end{theorem}

\begin{proof}
As outlined at the beginning of this section, we use Suurballe's algorithm to order the edges of $G$ by the lengths of their shortest cycles through the base edge. Then, using this order as a guide, we construct a greedy cycle sequence by repeatedly choosing an edge $f$ that is not part of the already-chosen cycles and using another instance of Suurballe's algorithm to find the shortest cycle through $f$ and the root edge, breaking ties in favor of cycles that use as few new edges as possible. By \autoref{lem:optimal4ambit}, the resulting set of cycles will form a minimum weight rooted cycle basis.

In a graph with $n$ vertices and $m$ edges, Suurballe's algorithm can be implemented in time $O(m+n\log n)$. The first stage of the algorithm may be implemented using Dijkstra's algorithm in this time bound. The second stage involves shortest paths in a graph $H$ with negative edge weights, to which Dijkstra's algorithm does not directly apply. However, in this second stage, we may re-weight each directed edge in $H$ from $u$ to $v$ with length $\ell$, giving it the new weight $\ell+d(s,v)-d(s,u)$, where $s$ is the starting vertex of the first path and $d$ is the shortest-path distance between two vertices in the input graph. This reweighting does not modify the comparison between any two path lengths, so the shortest paths in the reweighted version of~$H$ remain unchanged.  With these weights, the edges whose weights were negative become zero-weight, and all other edge weights remain non-negative, so Dijkstra's algorithm may again be applied.

A naive implementation of the algorithm applies Suurballe's algorithm $O(m)$ times so its total time is $O(m^2+mn\log n)$. However, this can be improved by observing that there are only $O(n)$ choices for the first path in Suurballe's algorithm, and that for each first path it is possible to handle all starting vertices of the second path, simultaneously, by using Dijkstra's algorithm to perform a single-destination shortest path computation. With this improvement the total runtime is $O(mn+n^2\log n)$. The algorithm as described so far applies only to unambiguously-weighted graphs but in Appendix~\ref{sec:unambiguous} we describe how to reduce the general problem to this case
in polynomial time.
\end{proof}

\paragraph{\bf Acknowledgements.}
The work of the first author was supported by the National Science Foundation under Grant CCF-1228639 and by the Office of Naval Research under Grant No. N00014-08-1-1015.

{\raggedright
\bibliographystyle{abuser}
\bibliography{bases}}

\begin{thebibliography}{10}

\bibitem{AguIst-JCB-12}
D.~Aguiar and S.~Istrail.
\newblock {HapCompass: A fast cycle basis algorithm for accurate haplotype
  assembly of sequence data}.
\newblock {\em J. Computational Biology} 19(6):577{--}590, 2012,
  \href{http://dx.doi.org/10.1089/cmb.2012.0084}%
{doi:10.1089/cmb.2012.0084}.

\bibitem{AmaIulRiz-IPCO-10}
E.~Amaldi, C.~Iuliano, and R.~Rizzi.
\newblock {Efficient deterministic algorithms for finding a minimum cycle basis
  in undirected graphs}.
\newblock {\em Integer Programming and Combinatorial Optimization: 14th
  International Conference, IPCO 2010, Lausanne, Switzerland, June 9-11, 2010,
  Proceedings}, pp.~397{--}410. Springer, Lecture Notes in Computer Science
  6080, 2010, \href{http://dx.doi.org/10.1007/978-3-642-13036-6\_30}%
{doi:10.1007/978-3-642-13036-6\_30}.

\bibitem{AndTho-JACM-07}
A.~Andersson and M.~Thorup.
\newblock {Dynamic ordered sets with exponential search trees}.
\newblock {\em J. ACM} 54(3):A13, 2007,
  \href{http://dx.doi.org/10.1145/1236457.1236460}%
{doi:10.1145/1236457.1236460}.

\bibitem{BouPetVil-PODC-04}
C.~Boulinier, F.~Petit, and V.~Villain.
\newblock {When graph theory helps self-stabilization}.
\newblock {\em Proc. 23rd ACM Symp. on Principles of Distributed Computing
  (PODC '04)}, pp.~150{--}159, 2004,
  \href{http://dx.doi.org/10.1145/1011767.1011790}%
{doi:10.1145/1011767.1011790}.

\bibitem{CouMakRot-TCS-00}
B.~Courcelle, J.~A. Makowsky, and U.~Rotics.
\newblock {Linear time solvable optimization problems on graphs on bounded
  clique width}.
\newblock {\em Theory of Computing Systems} 33(2):125{--}150, 2000,
  \href{http://dx.doi.org/10.1007/s002249910009}%
{doi:10.1007/s002249910009}.

\bibitem{ElkLieRiz-IPL-07}
M.~Elkin, C.~Liebchen, and R.~Rizzi.
\newblock {New length bounds for cycle bases}.
\newblock {\em Inform. Process. Lett.} 104(5):186{--}193, 2007,
  \href{http://dx.doi.org/10.1016/j.ipl.2007.06.013}%
{doi:10.1016/j.ipl.2007.06.013}.

\bibitem{Epp-SODA-03}
D.~Eppstein.
\newblock {Dynamic generators of topologically embedded graphs}.
\newblock {\em Proc. 14th Symp. Discrete Algorithms}, pp.~599{--}608. ACM and
  SIAM, 2003.

\bibitem{EppItaTam-Algs-92}
D.~Eppstein, G.~F. Italiano, R.~Tamassia, R.~E. Tarjan, J.~R. Westbrook, and
  M.~Yung.
\newblock {Maintenance of a minimum spanning forest in a dynamic plane graph}.
\newblock {\em J. Algorithms} 13(1):33{--}54, 1992,
  \href{http://dx.doi.org/10.1016/0196-6774(92)90004-V}%
{doi:10.1016/0196-6774(92)90004-V}.

\bibitem{GarJoh-79}
M.~R. Garey and D.~S. Johnson.
\newblock {\em {Computers and Intractibility: A Guide to the Theory of
  NP-Completeness}}.
\newblock W. H. Freeman, 1979.

\bibitem{Hor-SJC-87}
J.~D. Horton.
\newblock {A polynomial-time algorithm to find the shortest cycle basis of a
  graph}.
\newblock {\em SIAM J. Comput.} 16(2):358{--}366, 1987,
  \href{http://dx.doi.org/10.1137/0216026}%
{doi:10.1137/0216026}.

\bibitem{Kav-CMAME-76}
A.~Kaveh.
\newblock {Improved cycle bases for the flexibility analysis of structures}.
\newblock {\em Comput. Methods Appl. Mech. Engrg.} 9(3):267{--}272, 1976,
  \href{http://dx.doi.org/10.1016/0045-7825(76)90031-1}%
{doi:10.1016/0045-7825(76)90031-1}.

\bibitem{KavLieMeh-CSR-09}
T.~Kavitha, C.~Liebchen, K.~Mehlhorn, D.~Michail, R.~Rizzi, T.~Ueckerdt, and
  K.~A. Zweig.
\newblock {Cycle bases in graphs: Characterization, algorithms, complexity, and
  applications}.
\newblock {\em Comput. Sci. Rev.} 3(4):199{--}243, 2009,
  \href{http://dx.doi.org/10.1016/j.cosrev.2009.08.001}%
{doi:10.1016/j.cosrev.2009.08.001}.

\bibitem{KecKruHil-MSD-97}
A.~Kecskem{\'e}thy, T.~Krupp, and M.~Hiller.
\newblock {Symbolic processing of multiloop mechanism dynamics using
  closed-form kinematics solutions}.
\newblock {\em Multibody System Dynamics} 1(1):23{--}45, 1997,
  \href{http://dx.doi.org/10.1023/A:1009743909765}%
{doi:10.1023/A:1009743909765}.

\bibitem{LeeStr-DM-08}
A.~Lee and I.~Streinu.
\newblock {Pebble game algorithms and sparse graphs}.
\newblock {\em Discrete Math.} 308(8):1425{--}1437, 2008,
  \href{http://dx.doi.org/10.1016/j.disc.2007.07.104}%
{doi:10.1016/j.disc.2007.07.104}.

\bibitem{LemMaj-NAR-06}
S.~Lemieux and F.~Major.
\newblock {Automated extraction and classification of RNA tertiary structure
  cyclic motifs}.
\newblock {\em Nucleic Acids Research} 34(8):2340{--}2346, 2006,
  \href{http://dx.doi.org/10.1093/nar/gkl120}%
{doi:10.1093/nar/gkl120}.

\bibitem{Lie-ORP-07}
C.~Liebchen.
\newblock {Periodic timetable optimization in public transport}.
\newblock {\em Operations Research Proceedings} 2006:29{--}36, 2007,
  \href{http://dx.doi.org/10.1007/978-3-540-69995-8\_5}%
{doi:10.1007/978-3-540-69995-8\_5}.

\bibitem{Lov-FOCS-85}
L.~Lov{\'a}sz.
\newblock {Computing ears and branchings in parallel}.
\newblock {\em Proc. 26th Symp. Foundations of Computer Science (FOCS'85)},
  pp.~464{--}467, 1985, \href{http://dx.doi.org/10.1109/SFCS.1985.16}%
{doi:10.1109/SFCS.1985.16}.

\bibitem{MehMic-JEA-06}
K.~Mehlhorn and D.~Michail.
\newblock {Implementing minimum cycle basis algorithms}.
\newblock {\em ACM J. Exp. Algorithmics}~11, 2006,
  \href{http://dx.doi.org/10.1145/1187436.1216582}%
{doi:10.1145/1187436.1216582}.

\bibitem{ParMcCEpp-JMM-15}
B.~E. Parrish, J.~M. McCarty, and D.~Eppstein.
\newblock {Automated generation of linkage loop equations for planar one
  degree-of-freedom linkages, demonstrated up to 8-bar}.
\newblock {\em J. Mechanisms and Robotics} 7(1):011006, 2015,
  \href{http://dx.doi.org/10.1115/1.4029306}%
{doi:10.1115/1.4029306}.

\bibitem{Ram-SoPA-93}
V.~Ramachandran.
\newblock {Parallel open ear decomposition with applications to graph
  biconnectivity and triconnectivity}.
\newblock {\em Synthesis of Parallel Algorithms}, pp.~276{--}340. Morgan
  Kaufmann, 1993.

\bibitem{Riz-Algo-09}
R.~Rizzi.
\newblock {Minimum weakly fundamental cycle bases are hard to find}.
\newblock {\em Algorithmica} 53(3):402{--}424, 2009,
  \href{http://dx.doi.org/10.1007/s00453-007-9112-8}%
{doi:10.1007/s00453-007-9112-8}.

\bibitem{Suu-Nw-74}
J.~W. Suurballe.
\newblock {Disjoint paths in a network}.
\newblock {\em Networks} 4(2):125{--}145, 1974,
  \href{http://dx.doi.org/10.1002/net.3230040204}%
{doi:10.1002/net.3230040204}.

\bibitem{Tar-SJC-72}
R.~Tarjan.
\newblock {Depth-first search and linear graph algorithms}.
\newblock {\em SIAM J. Comput.} 1(2):146{--}160, 1972,
  \href{http://dx.doi.org/10.1137/0201010}%
{doi:10.1137/0201010}.

\bibitem{Tut-DM-71}
W.~T. Tutte.
\newblock {On the 2-factors of bicubic graphs}.
\newblock {\em Discrete Math.} 1(2):203{--}208, 1971,
  \href{http://dx.doi.org/10.1016/0012-365X(71)90027-6}%
{doi:10.1016/0012-365X(71)90027-6}.

\bibitem{Wel-MT-10}
D.~J.~A. Welsh.
\newblock {\em {Matroid Theory}}.
\newblock Dover, 2010, p.~131.

\bibitem{Whi-TAMS-32}
H.~Whitney.
\newblock {Non-separable and planar graphs}.
\newblock {\em Trans. Amer. Math. Soc.} 34(2):339{--}362, 1932,
  \href{http://dx.doi.org/10.2307/1989545}%
{doi:10.2307/1989545}.

\end{thebibliography}

\newpage
\appendix

\section{Suurballe's algorithm}
\label{sec:suurballe}

As it is most commonly described, Suurballe's algorithm~\cite{Suu-Nw-74} (a special case of the shortest-augmenting-path method for minimum cost flows) finds two edge-disjoint shortest paths between the same two vertices in a weighted graph. However, it can readily be adapted to find shortest vertex-disjoint paths from a single source vertex to two different vertices, to find the shortest rooted cycle containing a given vertex, or to find the shortest rooted cycle containing a given edge. As our proof depends on the details of this algorithm, we review them briefly here.

\begin{figure}[h]
\centering\includegraphics[height=2.5in]{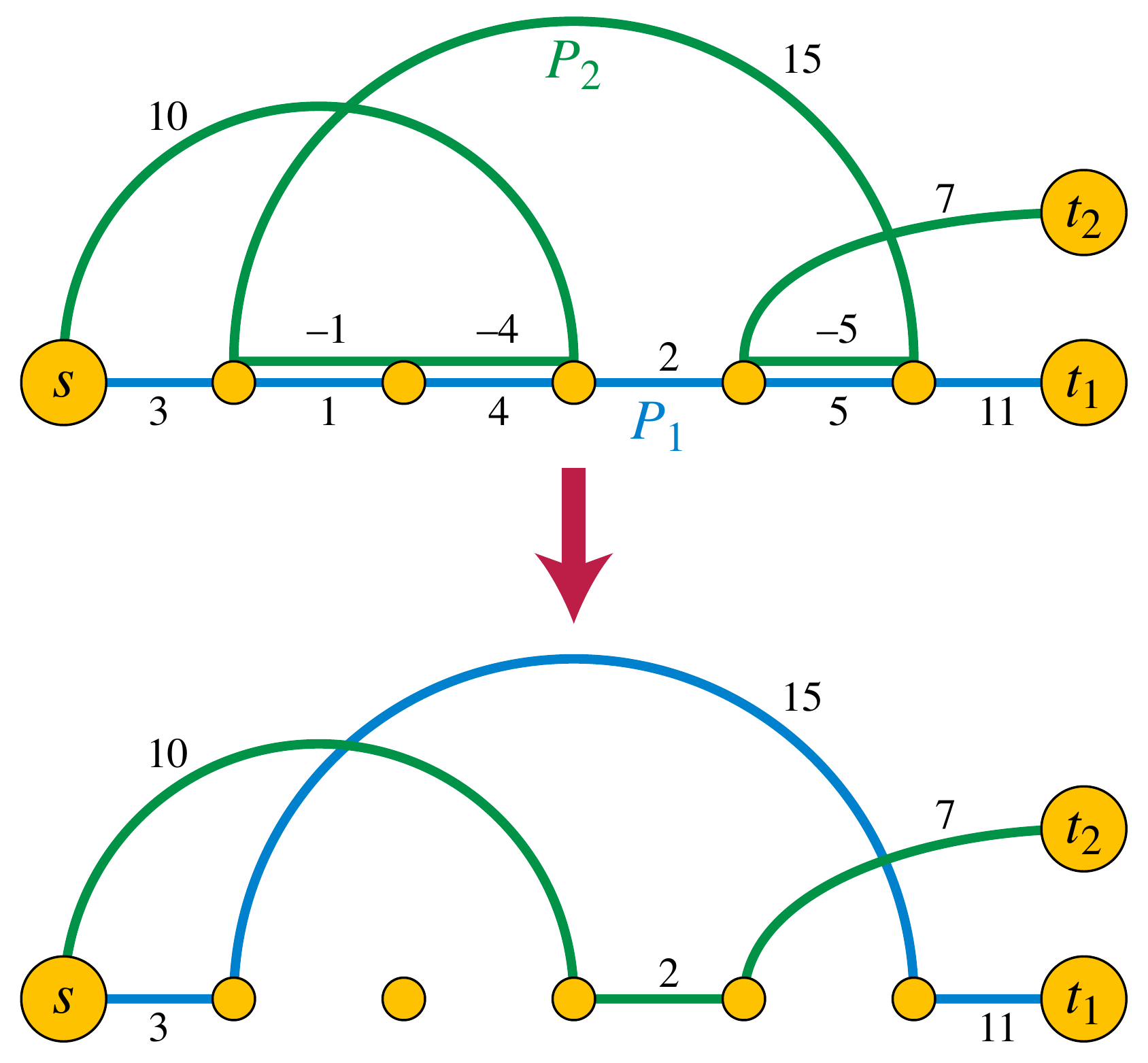}
\caption{Suurballe's algorithm: a shortest path $P_1$ from $s$ to $t_1$, together with a shortest path $P_2$ from $s$ to $t_2$ in a modified graph in which backwards travel along $P_1$ has negative cost, combine to form two disjoint paths of minimum total length from $s$ to $t_1$ and $t_2$.}
\label{fig:suurballe}
\end{figure}

We suppose that we are given a weighted directed graph $G$ with no negatively weighted cycles, a designated source vertex $s$, and two designated destination vertices $t_1$ and $t_2$. Suurballe's algorithm begins by finding a shortest path $P_1$ from $s$ to $t_1$ (in the graph $G\setminus\{t_2\}$, so that path $P_1$ avoids vertex $t_2$). It then constructs a new graph $H$, as follows:
\begin{itemize}
\item For each vertex $v$ of $G$, there are two vertices $v_{\text{in}}$ and $v_{\text{out}}$ in $H$.
\item If vertex $v$ does not belong to path $P_1$, then $H$ includes a weight-zero directed edge from  $v_{\text{in}}$ to $v_{\text{out}}$. If $v$ does belong to path $P_1$, then instead $H$ includes a weight-zero edge in the other direction, from $v_{\text{out}}$ to $v_{\text{in}}$.
\item For each directed edge $uv$ of $G$ that does not participate in path $P_1$, there is an edge with the same weight in $H$ from $u_{\text{out}}$ to  $v_{\text{in}}$.
\item For each directed edge $uv$ of $G$ that does  participate in path $P_1$, there is an edge with the negation of its weight in $H$ from $u_{\text{in}}$ to  $v_{\text{out}}$.
\end{itemize}

There can be no negative cycles in $H$, for any such cycle would either translate to a negative cycle in $G$ or could be used to improve path~$P_1$. Next, Suurballe's algorithm computes a shortest path in $H$ from $s$ to $t_2$. In terms of the original graph $G$, this corresponds to a path $P_2$  from $s$ to $t_2$ that travels forwards along edges disjoint from $P_1$, backwards along edges of $P_1$, and cannot touch a vertex of $P_1$ without traveling along an edge incident to that vertex. Additionally (with the additional assumption that $G$ has no zero-length cycles) the maximal subsets of contiguous edges of $P_1$ used by $P_2$ appear in the same order in both paths, for a violation of this ordering would  provide a cycle along which one of the two paths could be improved.

The symmetric difference of the two paths $P_1$ and $P_2$ (the set of edges used in one but not both of them) forms the union of the two desired vertex-disjoint paths from $s$ to $t_1$ and $t_2$, as shown in \autoref{fig:suurballe}.

To find the shortest path from $t_1$ to $t_2$ through an edge $e$ rather than a source vertex $s$, we may apply the same algorithm, choosing a shortest path from one endpoint of $e$ to $t_1$ in $G$ and then either from the other endpoint to $t_2$ (if the first path does not pass through $e$) or from the same endpoint (if it does pass through~$e$).

\newpage

\section{Making the weights unambiguous}
\label{sec:unambiguous}

The basic version of our minimum-weight basis algorithm, described in the main text of the paper, requires a graph that is \emph{unambiguously weighted}: no two paths or cycles have the same total weight. We describe here two techniques to simulate an unambiguous weighting, for a graph that is not already unambiguously weighted: one randomized without any slowdown and one that is deterministic but slower.

We assume a model of computation in which, however the given edge weights are represented,
comparisons are exact: any comparison of two sums of weights (of paths or cycles in the graph) is computed correctly and deterministically. We assume that sums of polynomial numbers of input values can be computed and compared exactly, and that our model of computation is capable of representing integers with $O(\log n)$ bits of precision (else how could it address the memory cells required to store the input) but we do not allow unlimited-precision arithmetic beyond these assumptions.

Our task is to augment the weight comparisons of the cycle basis algorithm with extra information so that, whenever two different sets of edges have the same weight, we can nevertheless determine an unambiguous perturbed comparison result. This perturbation should not change any weight comparisons that were already unequal in the given input, it should be performed efficiently, and it should not generate any inconsistencies (such as cyclic order relations) that could cause our algorithms to produce invalid results.

\subsection{Randomized}

Note that, to perform our algorithm correctly, we do not actually need a truly unambiguous weighting. It is sufficient to achieve a weighting in which all comparisons of weight sums performed by the algorithm are unambiguous. For, if we find a perturbed set of weights $w$ for which all algorithmic comparisons are unequal, then there also exists a perturbation of the perturbation, $w'$, such that $w'$ is truly unambiguous and such that the algorithm makes the same choices on $w'$ as it makes on $w$. But, on $w'$, the algorithm produces a correct output, so it must also be correct on~$w$.

Randomly perturbing the weights of a graph in such a way that our algorithm sees only unambiguous weight comparisons, with high probability, is relatively straightforward. For each edge $e$, we choose a random $c\log n$-bit integer $r_e$ for an appropriate constant~$c$. Then, whenever we compute a sum of edge weights, we also compute along with it the sum of these random integers. If the algorithm compares two sums  of weights and finds them to be equal, it then compares the corresponding sums of the $r_i$, and uses the result of that comparison. If the sums of the $r_i$ are also equal, then we choose arbitrarily which of the two values to consider as smaller.

Any individual comparison of the weights of two distinct edge sets has probability at most $2^{-c\log n}=n^{-c}$ of finding equal sums of $r_i$. Since the algorithm makes at most $O(n^3)$ comparisons in total (in the case of dense graphs), any choice of $c>3$ will cause the algorithm to succeed with high probability, by the union bound.
After the algorithm has run to completion, we verify that the result has the property that the newly-covered edge sets of each cycle form an ear decomposition with the correct number of ears. If so, it must be a valid minimum-weight rooted cycle cover, and if not, a failure must have occurred so we choose a new random perturbation and try again.

We summarize the results of this section:

\begin{theorem}
We can find a minimum-weight cycle basis for graphs with $n$ vertices and $m$ edges
in randomized expected time $O(mn+n^2\log n)$.
\end{theorem}

\subsection{Deterministic}

An unambiguous perturbation can always be achieved by choosing an arbitrary distinct positive integer index~$i$ for each edge, choosing a sufficiently small value of $\delta$ (smaller than the minimum difference between the weights of two edge sets, and adding $\delta/2^i$ to the weight of an edge with index~$i$. In this way, all perturbations are small enough and no two perturbations are equal. It is not necessary with this approach to determine an explicitly numeric value for $\delta$; instead, one can compare the unperturbed weights of edge sets and, if the result is equal, compare the perturbation values, as we did in the randomized algorithm. However, computing the perturbation values numerically would seem to require $n$-bit numerical precision, not allowed in our computational model. Instead, we describe data structural techniques that can determine the same comparison values, deterministically, in polylogarithmic time per comparison.

We define the \emph{incremental first difference problem} as follows. We have an ordered collection of elements (in the rooted cycle basis problem, these are the edges of the input graph) and we wish to form sets of elements, starting from the empty set, either by adding one element at a time or by specifying the elements that belong to a set (in the rooted cycle basis problem, these sets are paths and cycles in the graph). Additionally, we wish to be able to compare two sets to determine whether they are equal or unequal, and if they are unequal, to find the element at which they differ that is earliest in the element ordering.

To solve this problem, we fix a balanced binary tree $T$ whose leaf nodes are the elements and whose leaf ordering is the same as the ordering of the elements. For node $j$ of the tree, we let $L_j$ denote the left child of~$j$ (another node or a leaf element), we let $R_j$ denote the right child of~$j$, and we let $D_j$ denote the set of leaf elements descending from $j$.

\begin{figure}[b]
\centering\includegraphics[scale=0.4]{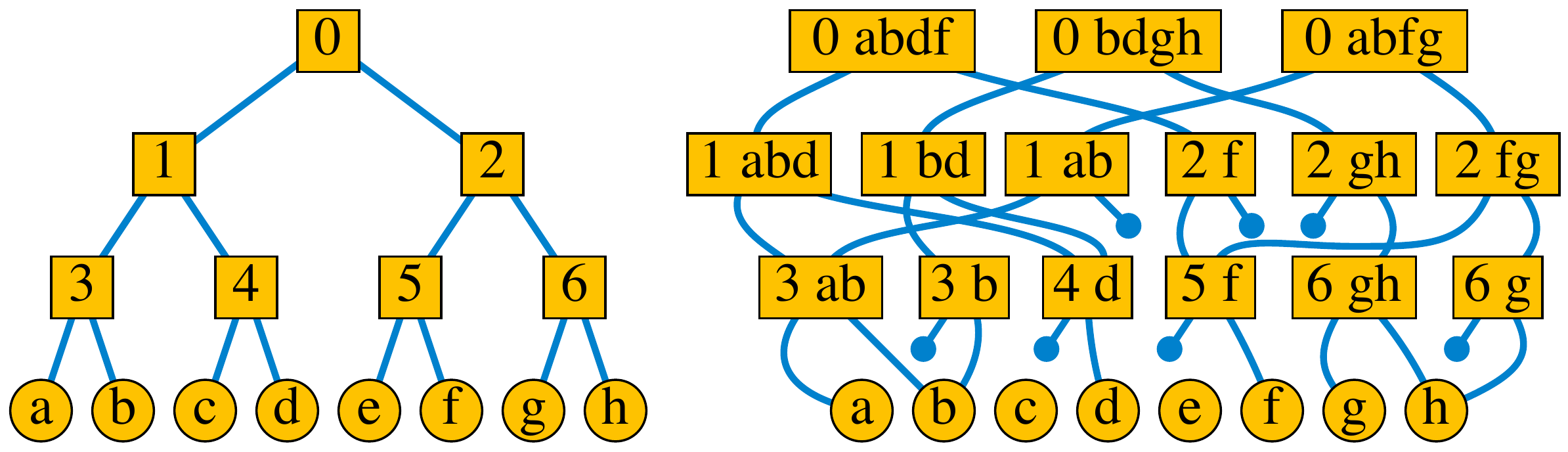}
\caption{The binary tree $T$ of ordered elements and the subset objects of an incremental first difference data structure for the three sets $\{a,b,d,f\}$, $\{b,d,g,h\}$, and $\{a,b,f,g\}$. The blue lines show the connections between each tree node or subset object and its two children. The set labels on the subset objects are not part of the data structure. The subset dictionary is not shown.}
\label{fig:set-comparison}
\end{figure}

Our representation will consist of a collection of \emph{subset objects},
each of which is associated with one particular node $j$ of $T$ and represents a nonempty subset of $D_j$. If a subset object $x$ is associated with tree node $j$ and represents the set $X$, then $x$ stores pointers for two child subset objects: the left child $\ell_x$, associated with tree node $L_j$ having set $X\cap D_{L_j}$, and the right child $r_x$ associated with tree node $R_j$ having set $X\cap D_{R_j}$. If either of these two intersections is empty then the corresponding child pointers $\ell_x$ or $r_x$ may be set to null pointers.
In this way, one of the given sets in the incremental first difference problem may be represented directly as a subset object associated with the root of~$T$, and by the binary tree of subset objects reachable from that root object.

In order to compare sets efficiently, we require that no two subset objects associated with the same node $j$ can have the same associated set. If two of the given sets have the same intersection with $D_j$, then they both share the same subset object for that tree node. Thus, we need to be able to find and re-use existing subset objects when necessary. To do so, we use a deterministic dictionary data structure that we call the \emph{subset dictionary}. Its keys are the triples $(j,\ell_x,r_x)$ associated with each subset object, and the value associated with such a key is the subset object having that triple of values. The subset dictionary could be implemented as a binary search tree, with logarithmic cost per operation; however, for integer keys with polynomial range, more efficient (in theory) deterministic dictionary structures are available with time $O((\log\log n)^2/\log\log\log n)$ per operation~\cite{AndTho-JACM-07}. To transform a triple $(j,\ell_x,r_x)$ into an integer with polynomial range, we assign each subset object a unique index number, we let $q$ be a polynomial bound on the number of subset objects to be created, and we calculate an integer key as $jq^2+aq+b$ where $a$ and $b$ are the index numbers of $\ell_x$ and $r_x$.

An example of the shared subset objects for three sets is illustrated in \autoref{fig:set-comparison}.

\begin{lemma}
\label{lem:set-comparison}
Given two sets with this representation, it is possible to determine the smallest element at which the two sets differ in time $O(\log n)$.
\end{lemma}

\begin{proof}
We begin with the two subset objects $x$ and $y$ associated with the root node of~$T$ that represent the two sets. If these are the same object, then the two sets are the same; otherwise, they must either have differing left or right child pointers. If the left children differ, we replace $x$ by $\ell_x$ and $y$ by $\ell_y$, and otherwise we replace $x$ by $r_x$ and $y$ by $r_y$, in either case moving down one level in~$T$ to another pair of differing sets containing the same smallest difference. Eventually we will reach a state where either $x$ or $y$ is a null pointer, after which we can follow the leftmost non-null path downward from the remaining non-null subset object to find the smallest differing element. This process takes constant time per level of~$T$ so its total time is $O(\log n)$.
\end{proof}

\begin{lemma}
\label{lem:incremental-set-comparison}
Given a set $S$ with this representation, and an element $e$, we can add $S\cup\{e\}$ to the representation in time $O(\log n(\log\log n)^2/\log\log\log n)$.
\end{lemma}

\begin{proof}
More generally, we show that for any subset object $x$, associated with a tree node $j$ and with associated set $X$, we can find or construct a subset object for $X\cup \{e\}$. To do so,
we determine whether $e$ is a left or right descendant of~$j$, and let $y$ be the left or right child of $x$, associated with a set $Y$. We recursively find or construct a subset object $y'$ for the set $Y\cup\{e\}$, then look up the triple of $j$, $y'$, and the other child of $x$ in the subset dictionary. If this lookup succeeds, we have found the subset object we desire. If it fails, we create a new subset object. In this way, whenever we create a new object we can be sure that its associated tree node and set do not coincide with any subset object that has already been created.

This construction process takes a constant number of dictionary operations per level of~$T$, or a logarithmic number of dictionary operations overall overall.
\end{proof}

Using the same bottom-up construction process, looking up each new triple in the subset dictionary prior to constructing a new subset object for that triple, we can show:

\begin{lemma}
\label{lem:all-at-once-set-comparison}
Given a set $S$ as a sorted list of its elements, and an incremental first difference data structure on a universe of $n$ elements, we can include $S$ in the sets represented by the data structure in time $O(n(\log\log n)^2/\log\log\log n)$.
\end{lemma}

Using this data structure to disambiguate equal weight sums, we have:

\begin{theorem}
We can find a minimum-weight cycle basis for graphs with $n$ vertices and $m$ edges
deterministically in time $O((mn\log n(\log\log n)^2/\log\log\log n+n^2\log^2 n)$.
\end{theorem}

\begin{proof}
We build an incremental first difference data structure on the edge sets of the paths and cycles created by the algorithm for unambiguously-weighted graphs. This data structure slows down the instances of Dijkstra's and Suurballe's algorithms within the overall algorithm by the time bound given in \autoref{lem:incremental-set-comparison}, as these algorithms build up their paths by one edge at a time allowing \autoref{lem:incremental-set-comparison} to be used. Once Suurballe's algorithm has found two paths whose symmetric difference forms a minimum-weight cycle, we can construct a set for that cycle using \autoref{lem:all-at-once-set-comparison}.

Within these algorithms, whenever a weight comparison determines that two edge sets have the same weight sum, we use \autoref{lem:set-comparison} to find the edge making the dominant contribution to the perturbations for those edge sets, and use this information to break the tie appropriately. This modification slows each comparison down by a logarithmic factor. In particular,  maintaining the priority queue in Dijkstra's algorithm only involves weight comparisons, and not construction of any new edge sets, so the $O(n^2\log n)$ term of our unambiguously-weighted algorithm is only slowed down by this factor rather than by the slightly larger factor of the $O(mn)$ term.
\end{proof}

\newpage

\section{Fundamental rooted cycle bases}
\label{sec:duality}

A \emph{fundamental} rooted cycle basis is a rooted cycle basis that is generated by a spanning tree, in the sense that all cycles in the basis are formed by a path in the spanning tree together with a single non-tree edge. As we show, it is NP-hard but fixed-parameter tractable to determine whether such a basis exists.

\subsection{Hardness}

We define a \emph{plane graph} to be a connected graph embedded in the plane with no edge crossings; that is, it is a planar graph together with a choice of how to embed the graph. The \emph{dual graph} of a plane graph $G$ has a vertex for each bounded or unbounded face of the embedding of $G$; each edge in $G$ separates two faces (or possibly a single face from itself) and corresponds to a dual edge connecting the corresponding two vertices. The dual graph is also plane but not necessarily simple: it may have multiple adjacencies (more than one edge between the same two vertices) and self-loops (an edge with both endpoints equal to each other). It has a planar embedding for which each dual vertex belongs to its corresponding face, and each edge of $G$ is crossed once by its corresponding dual edge and does not intersect any other dual edge. The dual graph of the dual graph of $G$ is isomorphic to~$G$ itself.

We define a \emph{rooted Hamiltonian cycle} of a graph $G$ rooted at an edge~$e$ to be a Hamiltonian cycle of $G$ that includes edge~$e$.

\begin{lemma}
\label{lem:duality}
Let $G$ be a plane graph rooted at $e$, and let $T$ be a spanning tree of $G$ containing $e$. Then the fundamental cycle basis for $T$ is rooted at $e$ if and only if the set of edges dual to the edges in $(G\setminus T)\cup\{e\}$ forms a Hamiltonian cycle for the dual graph of $G$.
\end{lemma}

\begin{figure}[b]
\centering\includegraphics[width=0.9\textwidth]{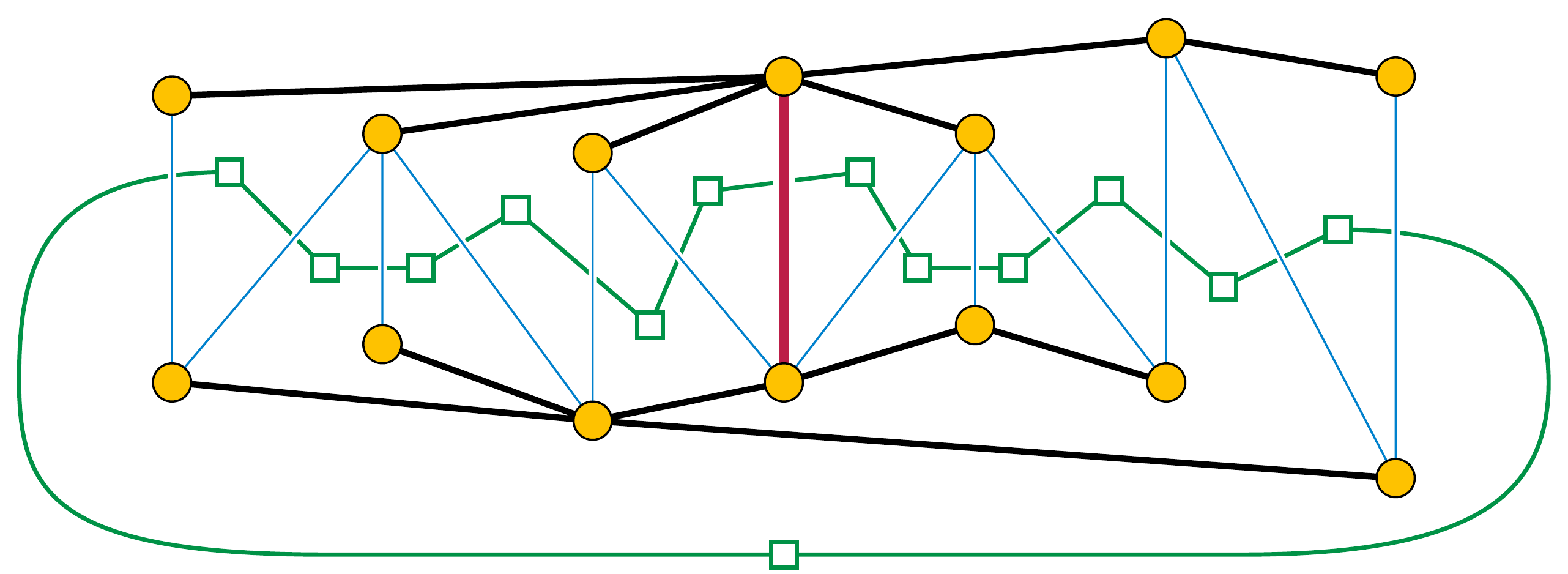}
\caption{A spanning tree generating a rooted fundamental cycle basis of a planar graph, and its dual Hamiltonian cycle}
\label{fig:dual-hamiltonian}
\end{figure}

\begin{proof}
Let $T'$ be the set of edges dual to the edges of $G\setminus T$; then $T'$ is necessarily a spanning tree of the dual graph of $G$, and $(T,T')$ form a \emph{tree-cotree decomposition} of~$G$~\cite{EppItaTam-Algs-92,Epp-SODA-03}. Let $P$ be the path in $T'$ connecting the two faces that are separated by~$e$, and let $e'$ be the edge dual to $e$.

If $P=T'$, then $T'\cup\{e'\}$ is a Hamiltonian cycle. By the Jordan curve theorem, this cycle separates the plane into two parts, an inside and an outside. Let $A$ and $B$ be the sets of vertices of $G$ inside and outside the cycle, respectively. Then $e$ and every edge of $G\setminus T$ crosses the cycle exactly once, so it has one endpoint in $A$ and the other endpoint in $B$. On the other hand, the edges of $T\setminus\{e\}$ do not cross the cycle at all, so both of their endpoints lie in the same set. Every fundamental cycle induced in $T$ by an edge $f\in G\setminus T$ crosses from $A$ to $B$ at $f$, and must return to $A$ on another edge, but the only other edge that passes from one set to another is~$e$. Therefore, every fundamental cycle contains~$e$, and the fundamental cycle basis is rooted.

If $P\ne T'$, then there exists a dual edge $f'\in T'$ that is adjacent to a vertex $v$ of~$P$ but does not itself belong to~$P$.  $T'\cup\{e'\}$ is not a Hamiltonian cycle, because it has degree three or more at~$v$. Again, by the Jordan theorem, the cycle $P\cup\{e'\}$ separates the vertices into two subsets $A$ and $B$, with only $e$ and the edges dual to $P$ having endpoints in both subsets.
Let $f$ be the edge whose dual is $f'$. Then the fundamental cycle induced in~$T$ by~$f$ cannot include $e$, because if it did it would have only one edge with endpoints in both subsets, an impossibility. Therefore, the fundamental cycle basis of~$T$ includes a cycle that does not contain~$e$ and the basis is not rooted.
\end{proof}

\begin{corollary}
\label{cor:dual-ham}
Let $G$ be a plane graph rooted at~$e$. Then $G$ has a rooted fundamental cycle basis if and only if its dual graph, rooted at the dual edge~$e'$ of~$e$, has a rooted Hamiltonian cycle.
\end{corollary}

\begin{figure}[ht]
\centering\includegraphics[scale=0.5]{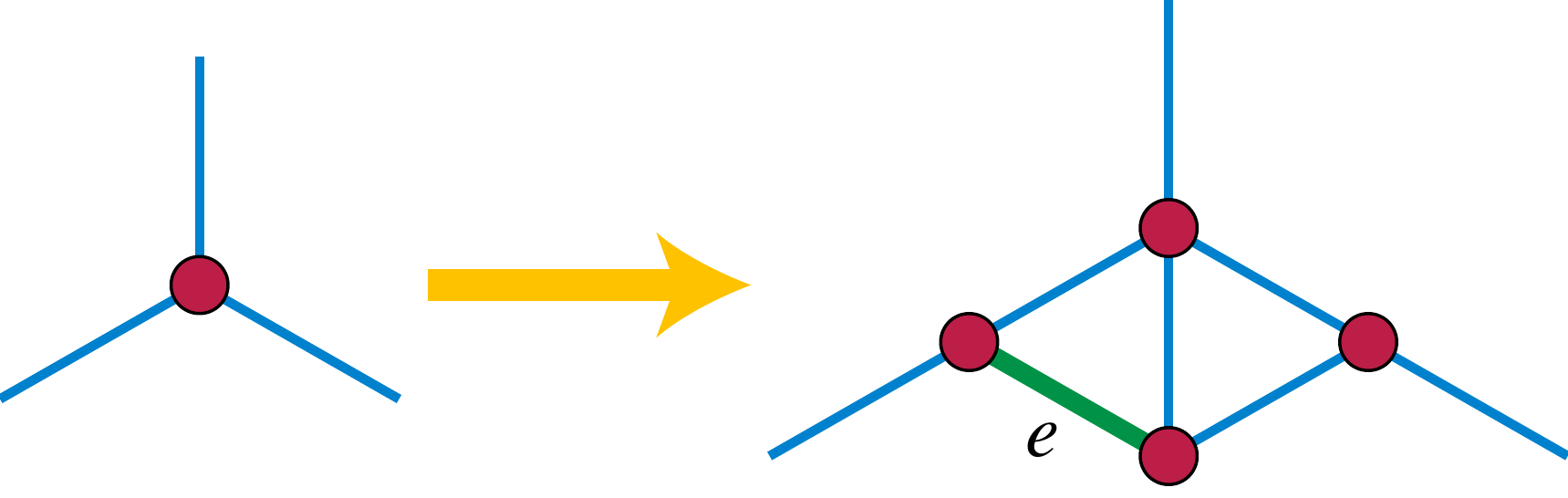}
\caption{If a degree-three vertex in a graph $G$ is replaced by the subgraph on the right, then $G$ has a Hamiltonian cycle if and only if the modified graph has a Hamiltonian cycle containing~$e$.}
\label{fig:forced-edge-gadget}
\end{figure}

\begin{lemma}
\label{lem:forced-edge-gadget}
Let $G$ be a graph containing a degree-three vertex, and replace this vertex by the four-vertex subgraph shown in \autoref{fig:forced-edge-gadget}. Then $G$ has a Hamiltonian cycle if and only if the modified graph has a Hamiltonian cycle that contains $e$.
\end{lemma}

\begin{proof}
If $G$ has a Hamiltonian cycle, it must use two of the three edges at the replaced vertex~$v$. A simple case analysis shows that in each of the three ways of choosing two edges, there exists a corresponding path through the four-vertex subgraph that uses the corresponding two edges, covers all vertices, and uses $e$. Therefore, if $G$ has a Hamiltonian cycle, the modified graph also has a Hamiltonian cycle that uses~$e$. In the other direction, if the modified graph has a Hamiltonian cycle that uses~$e$, then contracting the four-vertex subgraph to a single vertex produces a Hamiltonian cycle in~$G$.
\end{proof}

\begin{lemma}
\label{lem:planar-forced-edge-hard}
The problem of determining whether a given 3-connected rooted planar graph contains a rooted Hamiltonian cycle is \NP-complete.
\end{lemma}

\begin{proof}
We prove the result by means of a reduction from the known \NP-complete problem of finding a Hamiltonian cycle in a 3-regular 3-vertex-connected planar graph~\cite{GarJoh-79}. Given such a graph $G$, modify it according to \autoref{lem:forced-edge-gadget}, giving a graph $H$ and an edge $e$ such that we wish to find a Hamiltonian cycle in $H$ through~$e$. 
$G$ has a Hamiltonian cycle if and only if $H$ has a Hamiltonian cycle through~$e$, and the transformation preserves planarity and 3-connectivity.
\end{proof}

\begin{theorem}
It is \NP-complete to determine whether a given rooted graph has a rooted fundamental cycle basis.
\end{theorem}

\begin{proof}
We prove the result by means of a reduction from finding a Hamiltonian cycle through a given edge~$e$ of a 3-connected planar graph~$G$ (\autoref{lem:planar-forced-edge-hard}).
Embed $G$ planarly, let $H$ be its dual graph, and let $e'$ be the edge dual to~$e$. Because $G$ is 3-vertex-connected, $H$ is a simple graph.
Then by \autoref{cor:dual-ham}, $H$ has a fundamental cycle basis rooted at~$e'$ if and only if $G$ has a Hamiltonian cycle through~$e$. Since finding a rooted Hamiltonian cycle is \NP-hard, so is testing the existence of a rooted fundamental cycle basis.
On the other hand, it is straightforward to verify that a given spanning tree $T$ generates a rooted fundamental cycle basis, so testing the existence of a rooted fundamental cycle basis belongs to \NP. Since it is in \NP and is \NP-hard, it is \NP-complete.
\end{proof}

\subsection{Fixed-parameter tractability}

We now show that it is fixed-parameter tractable (parameterized by clique-width) to test for the existence of a fundamental cycle basis. This means that there is an algorithm for solving the problem whose running time is a polynomial of the input size (independent of the parameter) multiplied by a (non-polynomial) function of the parameter. The \emph{clique-width} parameter that we use is the minimum number of colors needed to construct the graph by a sequence of operations that take disjoint unions of colored graphs, add edges between all pairs of vertices matching a specified color pair, and recolor all vertices of a given color with a different color. Graphs of bounded treewidth also have bounded clique-width, but not necessarily vice versa. A \emph{monadic second-order} ($\MSO_1$) formula for a graph property is an expression that combines variables representing vertices or sets of vertices, quantifiers over these variables, terms that test whether a vertex belongs to a set or whether two vertices are adjacent, and the standard Boolean connectives over these terms. The connection between these two concepts is given by the following variation of Courcelle's theorem:

\begin{lemma}[Courcelle, Makowsky, and Rotics~\cite{CouMakRot-TCS-00}]
\label{lem:courcelle}
Any graph property that can be expressed by an $\MSO_1$ formula has a fixed-parameter tractable algorithm when parameterized by clique-width.
\end{lemma}

Thus, to prove fixed-parameter-tractability, we need only find an $\MSO_1$ formula expressing the existence of a fundamental rooted cycle basis. The following lemma accomplishes this.

\begin{lemma}
\label{lem:mso}
Let $G$ be a graph rooted at edge $e$. Then $G$ has a fundamental rooted cycle basis if and only if there exist two vertex sets $S_1$ and $S_2$ such that all of the following conditions are met:
\begin{itemize}
\item $S_1$ and $S_2$ partition the vertices of $G$: no vertex belongs to both sets and every vertex belongs to at least one of the two sets.
\item The two endpoints of $e$ belong to different sets of the partition.
\item For each $i$, set $S_i$ induces a connected set. That is, there does not exist another set $C$ such that $S_i\cap C$ and $S_i\setminus C$ are both nonempty and such that there are no edges connecting $S_i\cap C$ to $S_i\setminus C$.
\item For each $i$, set $S_i$ is acyclic. This can be expressed by the property that $S_i$ induces a graph with an empty $2$-core: for each nonempty subset $D$ of $S_i$, there exists a vertex $v\in D$ that has at most one neighbor in $D$.
\end{itemize}
\end{lemma}

\begin{proof}
If $G$ has a rooted fundamental cycle basis generated by tree $T$, then let $T_1$ and $T_2$ be the two subtrees of $T$ formed by deleting edge $e$, and let $S_1$ and $S_2$ be the vertex sets of $T_1$ and $T_2$. Then these two sets partition the vertices, and induce connected subgraphs. The subgraphs they induce are also acyclic, because if the subgraph induced by $S_i$ included an edge $f$ that was not part of $T_i$, then $f$'s cycle in the fundamental cycle basis would not be rooted.

Conversely, suppose that $G$ and $e$ have sets $S_1$ and $S_2$ that meet these properties. Then $S_1$ and $S_2$ induce two disjoint subtrees of $G$ which together span all the vertices of $G$; adding $e$ to these two subtrees gives a spanning tree $T$. Every edge $f$ in $G\setminus T$ must connect one of these two subtrees to the other; the path in $T$ between the endpoints of $f$ must necessarily pass through $e$, because $e$ is the only edge in $T$ that spans the cut from $S_1$ to $S_2$. Thus, $T$ generates a rooted fundamental cycle basis.
\end{proof}

\begin{corollary}
\label{cor:mso}
The existence of a fundamental rooted cycle basis can be expressed by an $\MSO_1$ formula.
\end{corollary}

\begin{theorem}
\label{thm:fun-fpt}
Determining whether a given rooted graph has a rooted fundamental cycle basis is fixed-parameter tractable when parameterized by the clique-width of the graph.
\end{theorem}

\begin{proof}
This follows immediately from \autoref{lem:courcelle} and \autoref{cor:mso}.
\end{proof}

\end{document}